\documentclass[12pt,a4wide]{article}%

\usepackage{graphicx}
\usepackage{subfigure}
\usepackage{amsthm}
\usepackage{amsmath}
\usepackage{amsfonts}
\usepackage{amssymb}

\newcommand{\nn}{\nonumber}

\renewcommand{\Pr}{\mathbb{P}}

\newcommand{\Pc}[2]{\Pr \Big[#1\Big| #2\Big]}
\newcommand{\MP}[2]{\text{MP}(#1, #2)}
\renewcommand{\P}{\mathbb{P}}

\newcommand{\codi}{\Rightarrow}

\newcommand{\F}{\mathcal{F}}
\newcommand{\E}{\operatorname{E}}
\renewcommand{\P}{\mathbb{P}}
\newtheorem{theo}{Theorem}[section]
\newtheorem{lm}{Lemma}[section]
\newtheorem{df}{Definition}[section]

\newtheorem{rmk}{Remark}[section]

\newtheorem{prop}{Proposition}[section]
\newtheorem{example}{Example}[section]

\numberwithin{equation}{section}

\begin{document}

\begin{center}

{ \large \bf Growth of uniform infinite causal triangulations }

\vspace{30pt}

{\sl V.~Sisko}$\,^{a}$, {\sl A.~Yambartsev}$\,^{b}$ and
{\sl S.~Zohren}$\,^{c,d,b}$

\vspace{24pt}

{\footnotesize

$^a$~Department of Statistics, Universidade Federal Fluminense, Brazil

\vspace{10pt}

$^b$~Department of Statistics, University of S\~ao Paulo, Brazil

\vspace{10pt}

$^c$~Department of Physics, Pontifícia Universidade Cat\'olica do Rio de Janeiro, Brazil

\vspace{10pt}

$^d$~Rudolf Peierls Centre for Theoretical Physics, Oxford University, UK
}
\vspace{48pt}

\end{center}

\begin{abstract}
We introduce a growth process which samples sections of uniform infinite causal triangulations by elementary moves in which a single triangle is added. A relation to a random walk on the integer half line is shown. This relation is used to estimate the geodesic distance of a given triangle to the rooted boundary in terms of the time of the growth process and to determine from this the fractal dimension. Furthermore, convergence of the boundary process to a diffusion process is shown leading to an interesting duality relation between the growth process and a corresponding branching process. 
\end{abstract}
\newpage

\tableofcontents

\newpage

\section{Introduction}

In the field of quantum gravity, models of random geometry or sometimes called quantum geometry have been studied intensively in the search of a non-perturbative definition of the gravitational path integral (see \cite{Ambjorn:1997di} for an overview). In two dimensions one distinguishes between models of Euclidean quantum gravity, so-called dynamical triangulations (DT) \cite{Ambjorn:1997di} and Lorentzian quantum gravity, so-called causal dynamical triangulations (CDT) (\cite{Ambjorn:1998xu}, and \cite{Ambjorn:2009rv} for an overview of recent progress in two dimensions). 

In the area of probability theory, the uniform measure on infinite planar triangulations (UIPT) \cite{Angel:2002ta} has been introduced as a mathematically rigorous model of DT,\footnote{Lately also much progress has been made in understanding the scaling limit of DT as the Brownian map (see \cite{Clay-Le-Gall} for an overview).} while recently uniform infinite causal triangulations (UICT) \cite{Durhuus:2009sm,Sisko2011} have been employed as a mathematically rigorous model of CDT. In particular, the formulation of the UICT measure is based on a bijection to planar rooted trees as was first formulated in \cite{MYZ2001} and independently later in \cite{Durhuus:2009sm} (see Figure \ref{f0}). Both formulations are based on a similar bijection for the dual graphs of CDT which was introduced in \cite{DiFrancesco:1999em}. The work in \cite{Durhuus:2009sm} shows convergence of the uniform measure on causal triangulations in the limit where the number of triangles goes to infinity and proves that the fractal dimension is two almost surely (a.s.) as well as that the spectral dimension is bounded above by two a.s. In \cite{Sisko2011} further convergence properties of the UICT measure are proven, in particular, using the relation to a size-biased critical Galton-Watson process, the convergence of the joint boundary length-area process to a diffusion process is shown from which one can extract the quantum Hamiltonian through the standard Feynman-Kac procedure. In a different work \cite{anatoli}, the existence of a phase transition of the quenched Ising model coupled to UICT is shown. All the above mentioned articles rely on the bijection to trees and the relation to branching processes. In this article we give an alternative formulation of UICT through a growth process.

In \cite{Angel2003} Angel studied a growth process which samples sections of UIPT. This growth process is a mathematically rigorous formulation of the so-called peeling procedure for DT, as introduced by Watabiki in the physics literature \cite{Watabiki:1993ym}, where it can also be understood as a time-dependent version of the so-called loop equation, a combinatorial equation derived from random matrix models of DT \cite{Ambjorn:1997di}.

In the context of CDT a similar peeling procedure as for DT can be formulated as was shown recently \cite{CDTmatrix2}. Furthermore, one can relate it to a random matrix model which itself can be understood as a new continuum limit of the standard matrix model for DT \cite{Ambjorn:2008gk,Ambjorn:2008jf}.

In this article we introduce a growth process which samples sections of UICT by elementary moves in which a single triangle is added (see Figure \ref{fig1} and \ref{f2}). This growth process is based on the peeling procedure of CDT \cite{CDTmatrix2} and analogous to the corresponding growth process for UIPT \cite{Angel2003}. The growth process is related to a Markov chain $\{M_n\}_{n\geq0}$ with state space $\mathbb N=\{1,2,3,...\}$ which describes the evolution of the boundary length $M_n$ of the triangulation as a function of ``growth time'' $n$. Using this relation it is shown how to estimate the stoping times $n_t$ at which the growth process finishes a strip of a fixed geodesic distance $t$ to the rooted boundary. We use this to prove that the fractal dimension is almost surely two, in an alternative manner to the derivation using branching processes as was done in \cite{Durhuus:2009sm}. Further, we prove convergence of the Markov chain $\{M_n\}_{n\geq0.}$ to a diffusion process. It is then shown how to relate this diffusion process using a random time change to another diffusion process describing the evolution of the generation size of a critical Galton-Watson conditioned on non-extinction. This provides us with an interesting duality picture with the growth process on the one side and the branching process on the other side.

\section{A growth process for uniform infinite causal triangulations}\label{growthSec}

\subsection{Definitions}

We consider rooted causal triangulations of a cylinder $\mathcal C = S^1\times [1, \infty)$, where $S^1$ is the unit circle.

Consider a connected graph $G$ with a countable number of vertices embedded in $\mathcal C$. Suppose that all its faces
are triangles (using the convention that an edge incident to the same face on both sides counts twice, see \cite{Sisko2011} for more details). A triangulation $T$ of $\mathcal C$ is the pair of the embedded graph $G$ and the set $F$ of all the faces: $T = (G, F)$.

\begin{df} \label{defact}
A triangulation $T$ of $\mathcal C$ is called an {\it almost causal triangulation} (ACT) if the following conditions
hold:
\begin{itemize}
\item each triangular face of $T$ belongs to some strip $S^1 \times [j, j + 1], j =
1, 2, \dots,$ and has all vertices on the boundary
$(S^1\times \{j\}) \cup (S^1\times \{j+1\})$ of the strip $S^1\times [j, j + 1]$;

\item let $k_j = k_j(T)$ be the number of edges on $S^1\times \{j\}$, then we have
$0 < k_j < \infty$ for all $j = 1, 2, \dots$.
\end{itemize}
\end{df}

\begin{df}
A triangulation $T$ of $\mathcal C$ is called a {\it causal triangulation} (CT) if it is an almost causal (ACT) and any triangle has exactly one edge on the boundary of the strip to which it belongs.
\end{df}

\begin{example}
The first two sclices from the bottom of the triangulation of Figure \ref{f3} form a CT while the third strip is an example of an ACT.
\end{example}

\begin{df}
A triangulation $T$ of $\mathcal C$ is called {\it rooted} if it has a root. The
root in the triangulation $T$ consists of a triangular face $t$ of $T$, called the {\it root
triangle}, with an ordering on its vertices $(x, y, z)$. The vertex $x$ is the {\it root
vertex} and the directed edge $(x, y)$ is the {\it root edge}. The root vertex and
the root edge belong to $S^1\times \{1\}$.
\end{df}

\begin{df}
Two almost causal or two causal rooted triangulations of $\mathcal C$, say $T = (G, F)$ and
$T^\prime = (G^\prime, F^\prime)$, are equivalent if there exists a self-homeomorphism of $\mathcal C$ such
that it transforms each slice $S^1\times \{j\}, j = 1,\dots, M$ to itself preserving
its direction, it induces an isomorphism of the graphs $G$ and $G^\prime$ and a bijection
between $F$ and $F^\prime$, also the root of $T$ goes to the root of $T^\prime$.
\end{df}

We usually abbreviate ``equivalence class of embedded rooted (almost) causal triangulations"  by ``(almost) causal triangulations". In the same way we can define an (almost) causal triangulations of a cylinder $C_t = S^1\times [1,t]$, where $t=2,3, \dots$.

\subsection{Uniform infinite causal triangulations}

Denote by ${\mathcal{T}}(N,m_{0},m)$ the set of finite causal triangulation with $N$ triangles with a rooted boundary of length $m_0$ and second boundary of length $m$.
Let ${\mathcal{T}}(N,m_{0})$ be the set of finite causal triangulations with $N$
triangles, with the length of the rooted boundary equal $m_{0}$ and the
length of the other boundary not fixed, i.e.
\[
{\mathcal{T}}(N,m_{0})=\cup_{m=1}^{\infty}
{\mathcal{T}}(N,m_{0},m).
\]
Let $C(N,m_{0})=\#{\mathcal{T}}(N,m_{0})$ be the number of
triangulations of a cylinder with $N$ triangles and $m_{0}$
boundary edges of the rooted boundary. Define the uniform distribution on the (finite) set
${\mathcal{T}}(N,m_{0})$ by
\begin{equation}
\P_{N,m_{0}}(T)=\frac{1}{C(N,m_{0})}. \label{mera}%
\end{equation}

One can now define the limiting measure as the uniform measure on infinite causal triangulations, as done in \cite[Theorem 2]{Durhuus:2009sm} which is based on the generic random tree measure \cite[Theorem 2]{Durhuus:2006vk} (see also \cite{Sisko2011,Aldous1998}):

\begin{theo}
\bigskip\label{thconv}There exists the measure $\pi$, called the uniform infinite causal triangulation (UICT) measure, on the set of causal triangulations of the cylinder $\mathcal C$ such that
\[
\P_{N,m_{0}}\rightarrow\pi_{m_o},\quad
N\rightarrow\infty
\]
as a weak limit.
\end{theo}

There is an interesting relation between UICT and critical Galton-Watson family trees due to a bijection of causal triangulations and rooted planar trees (or forests) (see \cite{MYZ2001} and \cite{Durhuus:2009sm}) which is illustrated in Figure \ref{f0}. To get from the rooted causal triangulation to the rooted planar forest we remove all horizontal edges from the triangulation, furthermore at each vertex we remove the leftmost up-pointing edge. The result is a planar rooted forest where we chose the sequence of roots to be the sequence of vertices on the initial boundary starting with the root vertex $x$ of the root triangle $(x,y,z)$ and following the boundary in an anti-clockwise direction. Connecting the root vertices of the forests to a single external vertex one obtains a planar rooted tree. The inverse relation should now be clear. A detailed and slightly different formulation of the bijection can be found in \cite{MYZ2001,Durhuus:2009sm}.

Consider the critical Galton-Watson branching process with one particle type and off-spring distribution $p_{k}=(1/2)^{k+1},\ k\geq0.$ Let
$\eta_{j}$ be the number of particles at generation $j$. One easily sees that $\eta_{j}$ is a recurrent Markov chain with transition probability
\[
\P_{GW}(\eta_{j+1}=k\mid\eta_j=l) = \frac{1}{2^{k+l}}\binom{k+l-1}{k}.
\]

The following lemma clarifies the connection between the Galton-Watson branching process and the UICT (see for instance \cite[Lemma 4]{Durhuus:2006vk}):

\begin{lm}
\label{vetv} Recall the definition of $k_j$ in Definition \ref{defact}. We have
\begin{equation}
\pi_{m_{0}}(k_{j}=m)=\frac{m}{m_0}\P_{GW}(\eta_{j}=m\mid\eta_{0}=m_{0}). \label{sootn}%
\end{equation}
\end{lm}

Note that the RHS of the previous expression can also be interpreted as a critical Galton-Watson process conditioned on non-extinction, i.e. $\pi_{m_{0}}(k_{j}=m)=\lim_{N\to\infty} \P_{GW}(\eta_{j}=m\mid\eta_{0}=m_{0},\eta_{N}>0)$.

\begin{figure}[t]
\begin{center}
\includegraphics[width=12cm]{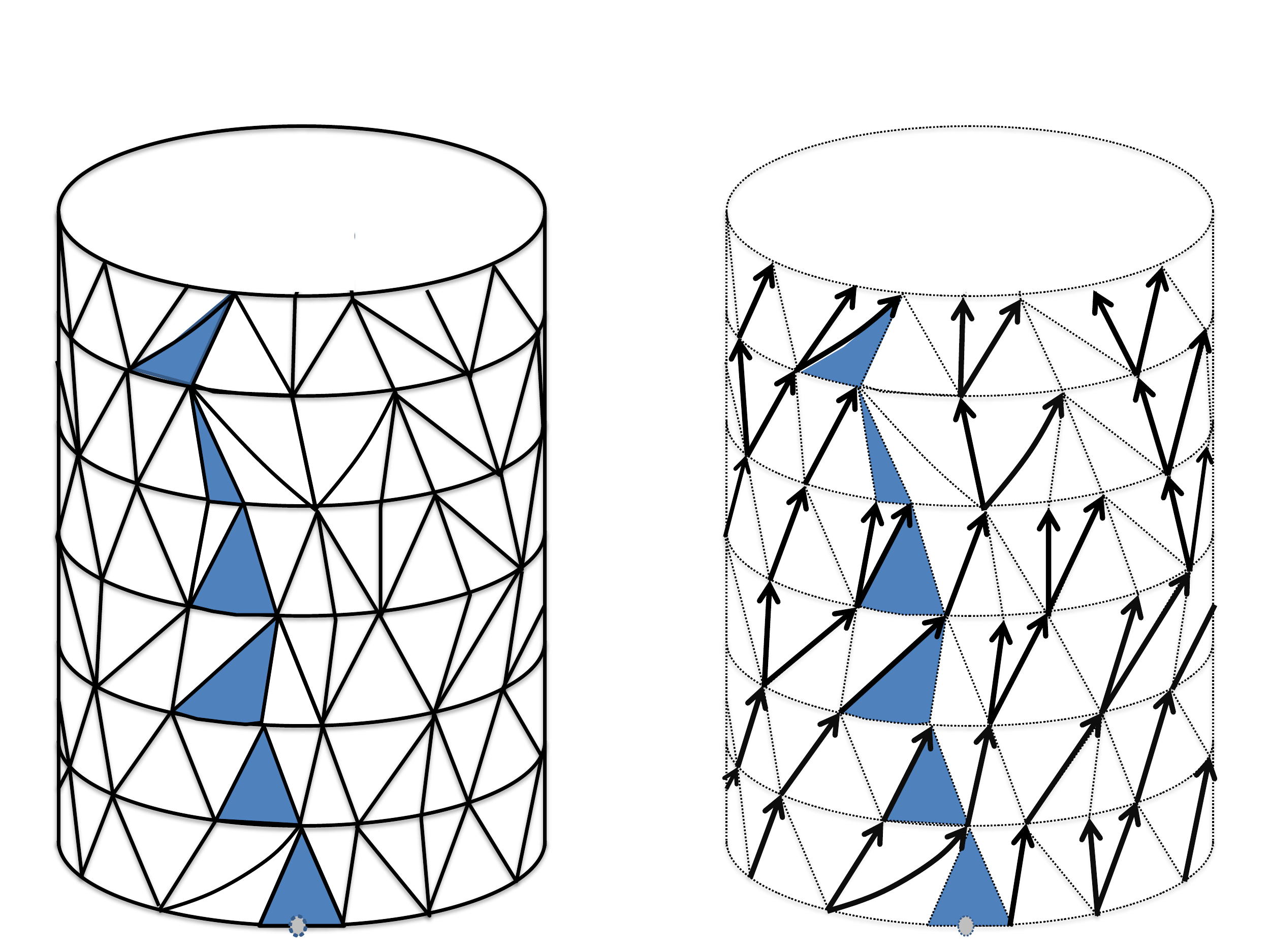}
\end{center}
\caption{Bijection between rooted causal triangulations and (rooted) forests.}%
\label{f0}%
\end{figure}

\subsection{Growth process}

Here we define the process with discrete time which constructs (samples) a UICT by adding one triangle at each step.

Let $\{x_1, \dots, x_m\}$ be the set of boundary vertices for some triangulation of the disc with $m$ boundary edges. We label the vertices following the direction on the boundary where $(x_m,x_1)$ is the marked edge.

In any step, we will add a triangle to the marked edge and after that we put the new mark on another edge. One allows this to be done in two different ways. In particular, one can add a triangle $(x_m,y, x_1)$ to the marked edge $(x_m,x_1)$ where $y$ is either a new vertex, we call this the $(+)$-move, or $y=x_2$, where $x_2$ is the next vertex after $x_1$ following the direction on the boundary; we call this the $(-)$-move. If $y$ is a new vertex, then the next marked edge will be $(y,x_1)$. In the case $y=x_2$ the new marked edge is $(x_m,x_2)$. If the boundary consists of only one edge $(x_1,x_1)$, then in the next step one can only add a triangle $(x_1, y, x_1)$ with the marked edge $(y,x_1)$. Note that the marked edge belongs to the boundary at each step of the growth process.

We will consider the following special starting triangulation with $m+1$ vertices $(m\ge 1)$: a triangulation of the disc with $m$ edges on the boundary (a $m$-gon), having $m$ triangles and one vertex in the interior of the disc which is a common vertex of all $m$ triangles. This vertex we call the  $0$-root or $0$ (in contrast with the root triangle). Let us denote this triangulation as $\mathcal S_m$. Note that any move preserves the topology of the triangulation as a disc. Denote by $T_n$ the triangulation of the disc after $n$ moves and let $l(T_n)$ be the length of the boundary of the triangulation $T_n$. Further, let $e_n=(v^1_n, v^2_n)$ be the marked edge of $T_n$ with vertices $v^1_n, v^2_n \in T_n$.

We now assign probabilities to the growth process: Conditioning on the length of the boundary of the triangulation $T_n$, we can add another triangle to it by choosing one of the above two moves randomly according to the probabilities
\begin{equation}\label{yamb-tp}
    \mathbb P( (\pm)-\mbox{move} \mid l(T_n)=m ) = \frac{1}{2} \frac{m\pm 1}{m}.
\end{equation}
Denote by $\mathcal T_m(n)$ the set of all possible triangulations of the disc obtained by applying all possible (permitted) sequences of length $n$ of the $(+)$ and $(-)$ moves starting with $\mathcal S_m$. Given the transition probabilities \eqref{yamb-tp} one has that $T_n$ is a Markov chain with  state space $\cup_{n\ge 0}\mathcal T_m(n)$.

Note that in any move one adds one triangle to the triangulation and changes the length of the boundary of the triangulation by one: the $(+)$-move increases the boundary by one, while the $(-)$-move decreases the boundary by one. This process of growing the triangulation is basically the time reversal of the so-called ``peeling" process, which is related to so-called loop equations for matrix models in the physics literature (see \cite{Ambjorn:1997di} in the context of DT and \cite{CDTmatrix2,Ambjorn:2008gk,Ambjorn:2008jf,ZohrenThesis} in the context of CDT).

\begin{figure}[t]
\centering
\subfigure[The $(+)$-move]{
\includegraphics[width=5cm]{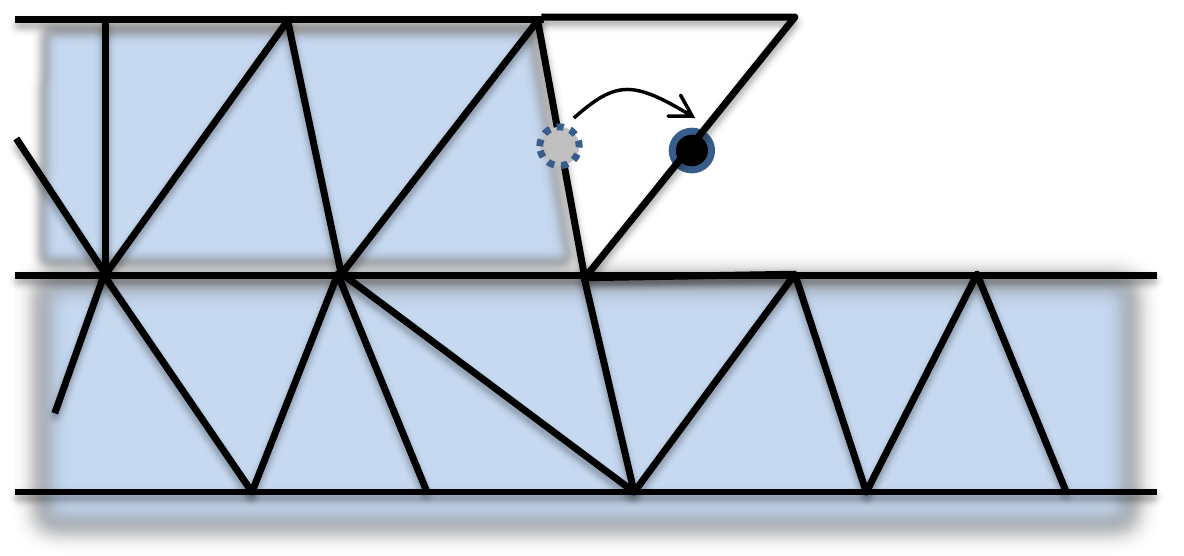}
\label{fig1a}
}
\subfigure[The $(-)$-move]{
\includegraphics[width=5cm]{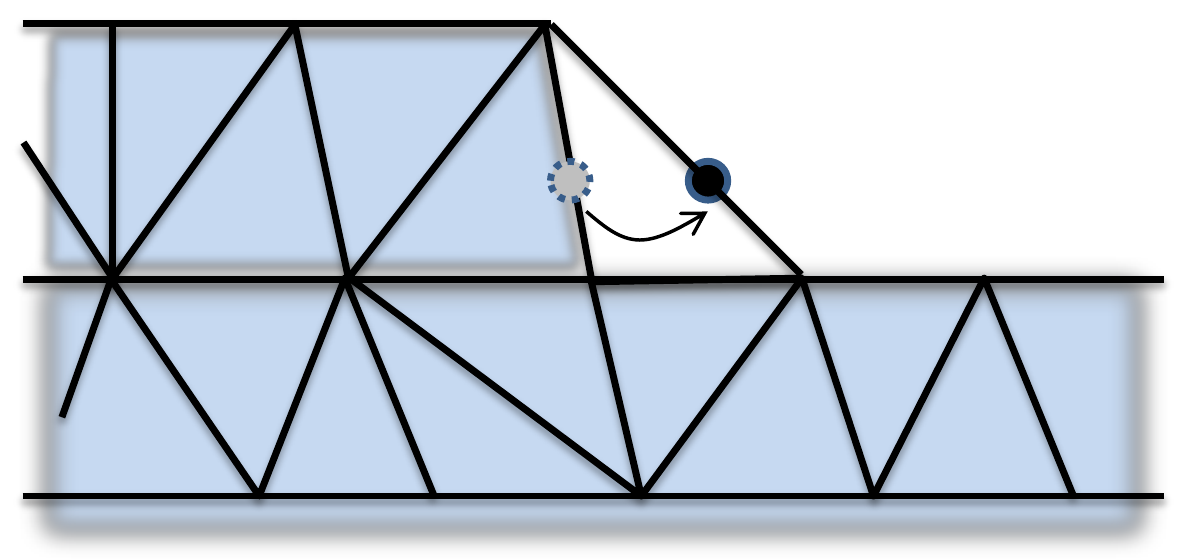}
\label{fig1b}
}
\caption{The two different moves of the growth process.\label{fig1}}
\end{figure}

The process $T_n$ determines a process which describes the evolution of the length of the boundary of the triangulation $T_n$. Denote the length $M_n=l(T_n)$. Define $\xi_n=M_{n+1}-M_n$. The probabilities \eqref{yamb-tp} one can rewrite as
\begin{equation}\label{yamb-tp1}
    \mathbb P( \xi_n = \pm 1 \mid M_n=m ) = \frac{1}{2} \frac{m\pm 1}{m}. 
\end{equation}
It is clear that $M_n$ is a Markov chain with state space  $\mathbb N=\{1,2,3,\dots\}$ and transition probabilities \eqref{yamb-tp1}.

We now describe the relationship between the process $T_n$ and the process $M_n$. If we know the sequence of $\{M_n\}_{n=0,1,\dots, s}$ then we know the sequence of $(\pm)$-moves and consequently we know the triangulation $T_s$. Inversely, if we fix the triangulation $T_s$ from $\mathcal T_m(s)$ we can reconstruct the sequence $\{M_n\}_{n=0,1,\dots, s}$. Hence, one has:

\begin{rmk}\label{yamb-th1}
For any $s\in \mathbb N$, there is a one-to-one correspondence $g$ between $\mathcal T_m(s)$ and the set of sequences $\{M_n\}_{n=0,1,\dots, s}$, with $M_0=m$, fulfilling
\begin{equation}\label{yamb-th1-p}
\mathbb P ( T_s = T \mid T_0 = \mathcal S_m) = \mathbb P ( g(T) \mid M_0=m ).\nn
\end{equation}
\end{rmk}


Due to this relation we also call the Markov chain $\{M_n\}_{n=0,1,\dots, s}$ the growth process.

We will now make the link to almost causal triangulations. For any $s\in \mathbb N$ and any triangulation $T_s$ from the set $\mathcal T_m(s)$ there exists a number $h(T_s)$ (to be defined as the ``height'' of the triangulation) such that the set of all vertices of the triangulation $V(T_s)$ can be divided into the disjoint sets corresponding to the distances between the vertices and the $0$-root: $V(T_s) = \cup_{i=0}^{h(T_s)} V_i$, where $V_i=V_i(T_s)$ is the set of vertices of $T_s$ which have distance to the $0$-root equal to $i$, where $V_0$ contains only the $0$-root vertex. Thus $h(T_s)$ is the maximal distance to the $0$-root.

\begin{df}\label{yamb-mom1}
For the growth process $T_n$ let us define the following moments $n_t$, $t=1, 2, \dots$.
\begin{equation}
    n_t:=\min\{ s>0:\ dist(v^1_s, 0\!-\!\text{root})=dist(v^2_s,  0\!-\!\text{root})=t \}\nn
\end{equation}
where we recall that $(v^1_s, v^2_s)=e_s$ are the vertices adjacent to the root edge.
\end{df}


Denote by $\hat{T}_n$  the triangulation $T_n$ without the $0$-root and the edges attached to it, then we have:

\begin{theo}\label{yamb-th2}
 $\hat{T}_{n_t}$ is an almost causal triangulation of $C_t$.
\end{theo}

This means that between the moments $n_t$ and $n_{t-1}$ the growth process $T_t$ constructs an almost causal  triangulation of the strip $S^1\times [t-1,t]$.

\begin{proof} The proof follows directly from the detailed description of the growth process: It is obvious that $\hat{T}_{n_t}$ is a triangulation of the cylinder $C_t$, because it is a triangulation of the disc without the faces of an initial $m$-gon $\mathcal S_m$. Removing $\mathbb{S}_m=\mathcal{S}_m\setminus \partial\mathcal{S}_m$ from the disc adds a hole in the disc and makes  $\hat{T}_{n_t}$ homeomorphic to the cylinder.

The theorem states that there exists a homeomorphism $f$ of the disc with a hole with the embedded graph $\hat{T}_{n_t}$ into the cylinder which maps the set $V_i$, $i=1,\dots,t$ into the slice $S^1\times\{i\}$ of the cylinder $C_t$ such that any triangle will belong to some strip $S^1\times [i,i+1]$. For that, firstly, we prove that for any $i$, $i=1,\dots,t$, there exists a Hamilton path (circle) consisting of all vertices $V_i$: suppose $|V_i|=l_i$, ordering the vertices of $V_i=\{x_1, x_2, \dots, x_{l_i}\}$ in order of their appearance we will show that there exists the circle $(x_j, x_{j+1}), j=1,\dots, l_i-1$ and $(x_{l_i},x_1)$ in $\hat T_{n_t}$. Secondly, we will show that the homeomorphism $f$ that maps the vertices $V_i$ with its Hamilton path into $S^1\times\{i\}$ maps any triangular face into one strip.

\begin{figure}[t]
\begin{center}
\includegraphics[width=8cm]{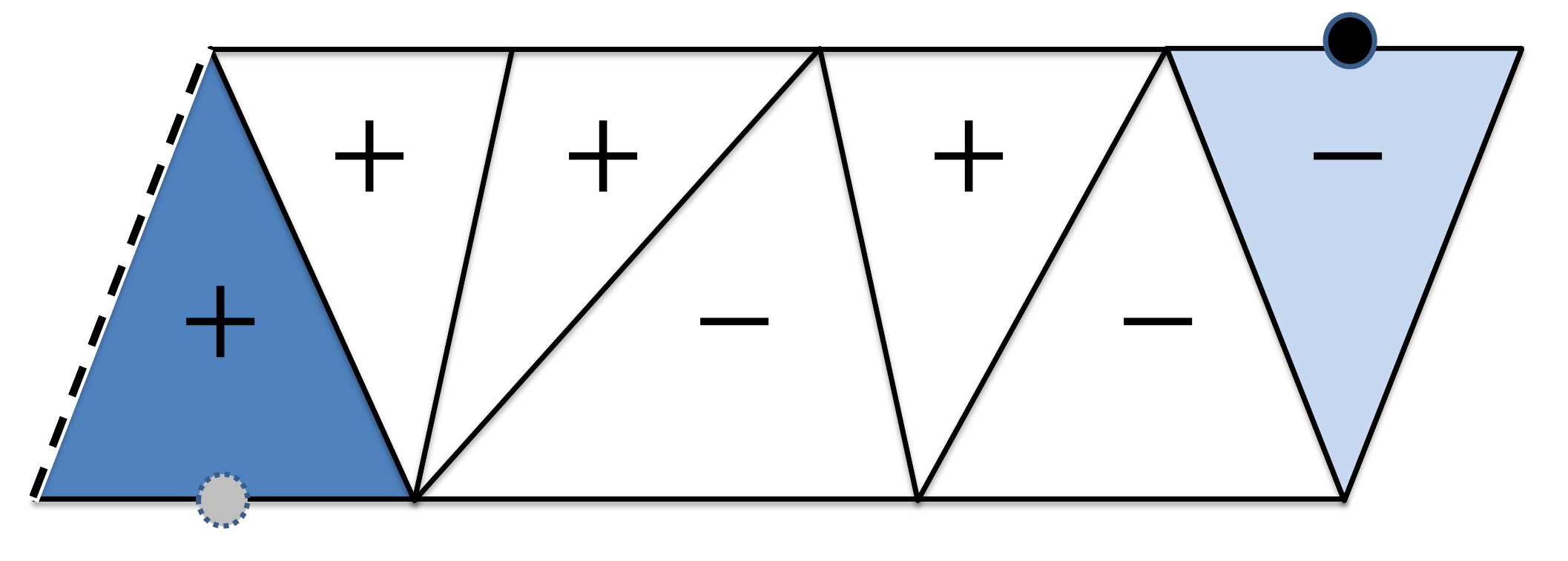}
\end{center}
\caption{Construction of a slice of an almost causal triangulation using the following sequence of moves: $+$, $+$, $+$, $-$, $+$, $-$, $-$. The left hand side and the right hand side of the strip are periodically identified.}%
\label{f2}
\end{figure}

In the following we describe the three phases of the construction of a triangulation of a strip $S^1\times [i-1,i]$: starting, filling and finishing off the strip.
\begin{enumerate}
\item[(i)]
{\it Starting a strip.} We start from $l_{i-1}$ edges $$E_{i-1} = \{ (x_1, x_2), \dots, (x_{l_{i-1}-1}, x_{l_{i-1}}), (x_{l_{i-1}}, x_{1})\}$$ and $l_{i-1}$ vertices $x_1, x_2, \dots, x_{l_{i-1}} \in S^1\times \{i-1\}$ with distance $i-1$ from the $0$-root: $V_{i-1}=\{x_1, x_2, \dots, x_{l_{i-1}} \}$. Let $e=(x_{l_{i-1}},x_1)$ be the marked edge.
The first phase continues until the first $(+)$-move.

Suppose $l_{i-1}>1$. If the first move is a $(+)$-move, then the process adds a new triangle $(x_{l_{i-1}},y_1,x_1)$ and $y_1\in S^1\times\{i\}$ has  distance $i$ to the $0$-root ($V_i=\{y_1\}$). The new marked edge $e=(y_1,x_1)$ connects two vertices with different distances to the $0$-root and we continue to the next phase. If, on the other hand, the first move is a $(-)$-move, it adds a new edge $(x_{l_{i-1}}, x_2)$ which will be the new marked edge. One observes that all points $x_1, \dots, x_{l_i}$ have the same distance to the $0$-root and the $(-)$-move does not change their distances. Moreover, a sequence of $k$ $(-)$-moves, $k< l_{i-1}$, maintains the set $V_i$, and the next $(+)$-move starts the next level. Thus starting with $k$ $(-)$-moves ($k< l_{i-1}$) before the first $(+)$-move one obtains the following boundary of the triangulation $(x_{l_{i-1}}, y_1,  x_{k+1}, \dots,  x_{l_{i-1}-1})$, with $y_1\in S^1\times \{ i\}$. Note that the case $x_{k+1}=x_{l_{i-1}}$ is allowed. In this particular case, after $l_{i-1}-1$ $(-)$-moves the length of the current boundary of the triangulation is equal 1, and the next step has to be a $(+)$-move.

In the case $l_{i-1} = 1$, $V_{i-1} = \{ x_1\}$, and $e=(x_1,x_1)$ the boundary of the triangulation is equal to 1, and the first step can be only a $(+)$-move.

\item[(ii)]
{\it Filling a strip}. During this phase we have a set of vertices $V_i=V_i(T_n)=\{y_1, \dots, y_k\}\ne \emptyset$ with distance $i, i\ge 2$ to the $0$-root in the graph $T_n$ and $h(T_n)=i$. Starting with $V_i=\{y_1\}$, let $E_{i-1}^0 \subseteq E_{i-1}$ be the set of edges which belong to the boundary of the triangulation (see Figure~\ref{f3}). The set $E_{i-1}^0$ decreases by one element with any $(-)$-move and once it becomes empty this phase stops and we proceed to $(iii)$. Furthermore, any $(+)$-move adds a new vertex to the set $V_i$ connected by an edge to the previous vertex from $V_i$.
%


\item[(iii)]
{\it Finishing off a strip}. This phase continues until the first $(-)$-move, which finishes a triangulation of the strip. Suppose that we have $l_i$ vertices in the set $V_i$ before the $(-)$-move. The marked edge connects $y_{l_i}$ (the last vertex in $V_i$) and $x_{l_{i-1}}$ (the last vertex in $V_{i-1}$): $e=(y_{l_i}, x_{l_{i-1}})$. The following edge on the boundary is $(x_{l_{i-1}}, y_1)$, thus the $(-)$-move will connect the vertices $y_{l_i}$ and $y_1$. Note also that this is the first moment when the next marked edge $e=(y_{l_i},y_1)$ will connect two points with the same distance $i$ to the $0$-root, i.e.\ it defines the moment $n_i$ as given in Definition~\ref{yamb-mom1}.
\end{enumerate}

From the description given above it is clear that any set $V_i=\{v_1, \dots, v_{l_i}\}$ of $T_n$ with $n\geq n_i$ has a circle connecting a sequence of vertices from $V_i$: the edges $(v_i,v_{i+1}), i=1,\dots, l_i-1$ are created by $(+)$-moves and $(v_{l_i}, v_1)$ is created by a $(-)$-move defining the moment $n_i$. Denote this circle graph by $G_i$. Moreover, any $(+)$-move during the filling stage will create a ``down" triangle of which exactly one edge will connect vertices with distance $i$ to the $0$-root and two edges will connect these two vertices to one vertex from $V_{i-1}$; further any $(-)$-move creates an ``up" triangle consisting of two vertices from $V_{i-1}$ and one vertex from $V_i$.

Thus the description of the construction of a strip provides the existence of a homeomorphism $f$ of the disc without $\mathbb S_m$ into $C_t$ such that the image of the $G_i$ in the disc maps into the circle $S^1\times \{i\}$ of $C_t$ for any $i=1,\dots, t$. Moreover any such homeomorphism maps the triangular faces of $T_n$ created between the time $n_{i-1}$ and $n_i$ into the strip $S^1\times [i-1,i]$ for $i=2,\dots, t$.
\end{proof}

\begin{rmk}\label{yamb-lemma1}
From the proof of the proceeding theorem it follows directly that $n_t <\infty$ $\mathbb P$-a.s.
\end{rmk}

\begin{figure}[t]
\begin{center}
\includegraphics[width=5cm]{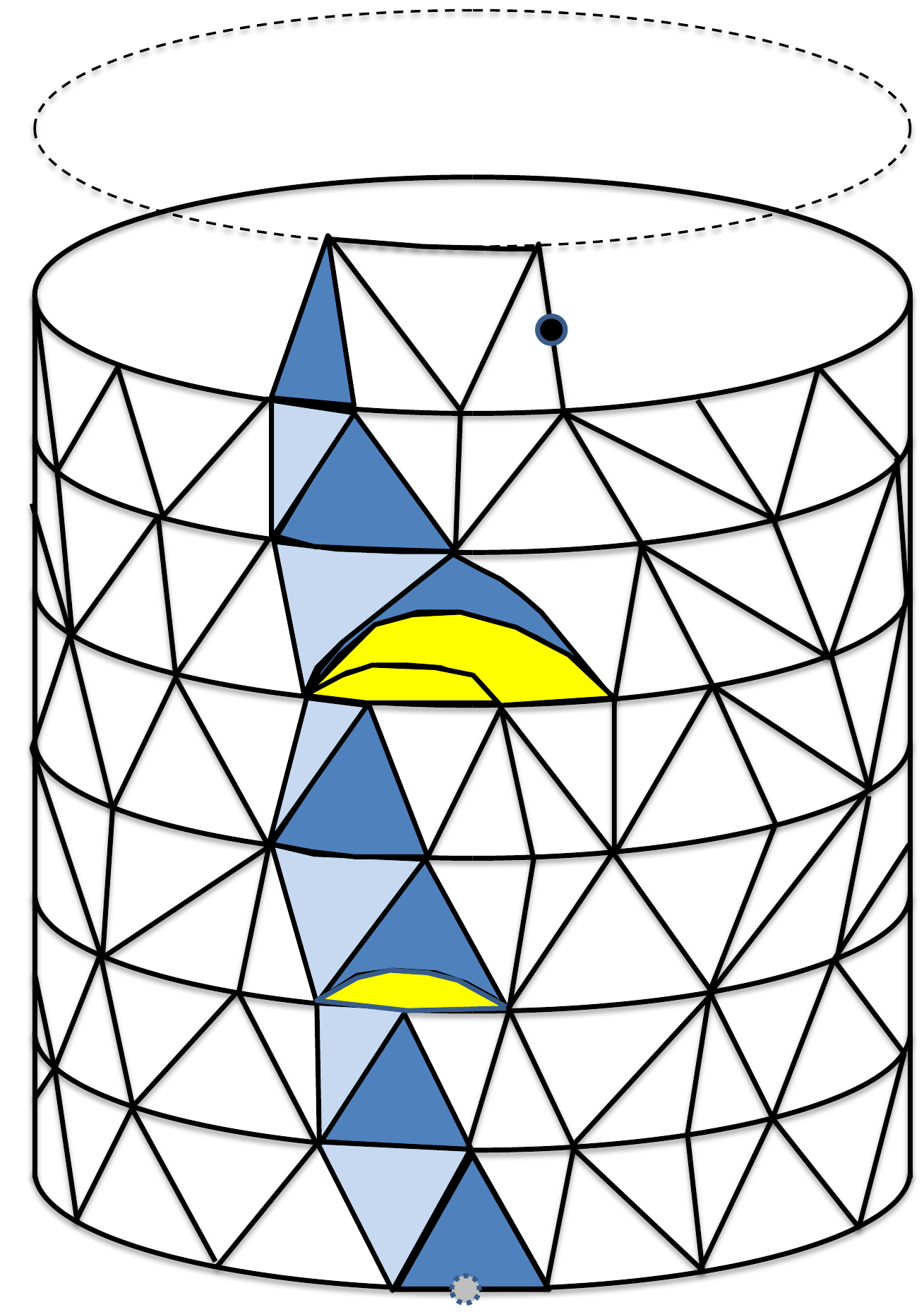}
\end{center}
\caption{An example of a triangulation created by the growth process, where the first six slices are an example of an almost causal triangulation. }%
\label{f3}
\end{figure}

\begin{example}
Figure \ref{f2} shows an example of a sequence of moves of size $n=7$: Starting with $\mathcal{S}_3$, i.e.\ $l(T_0)=M_0=3$ and creating a strip with final boundary of length $l(T_7)=M_7=4$. Here the last move completes the first strip of the triangulation, hence, $n_1=0$ and $n_2=7$. We observe that the result is a causal triangulation (with a specific, so-called staircase boundary condition). However, as one can observe in Figure \ref{f3}, if one starts a new strip with a $(-)$-move one can create certain outgrowths. Therefore the name almost causal triangulations.
\end{example}

\begin{figure}[t]
\begin{center}
\includegraphics[width=10cm]{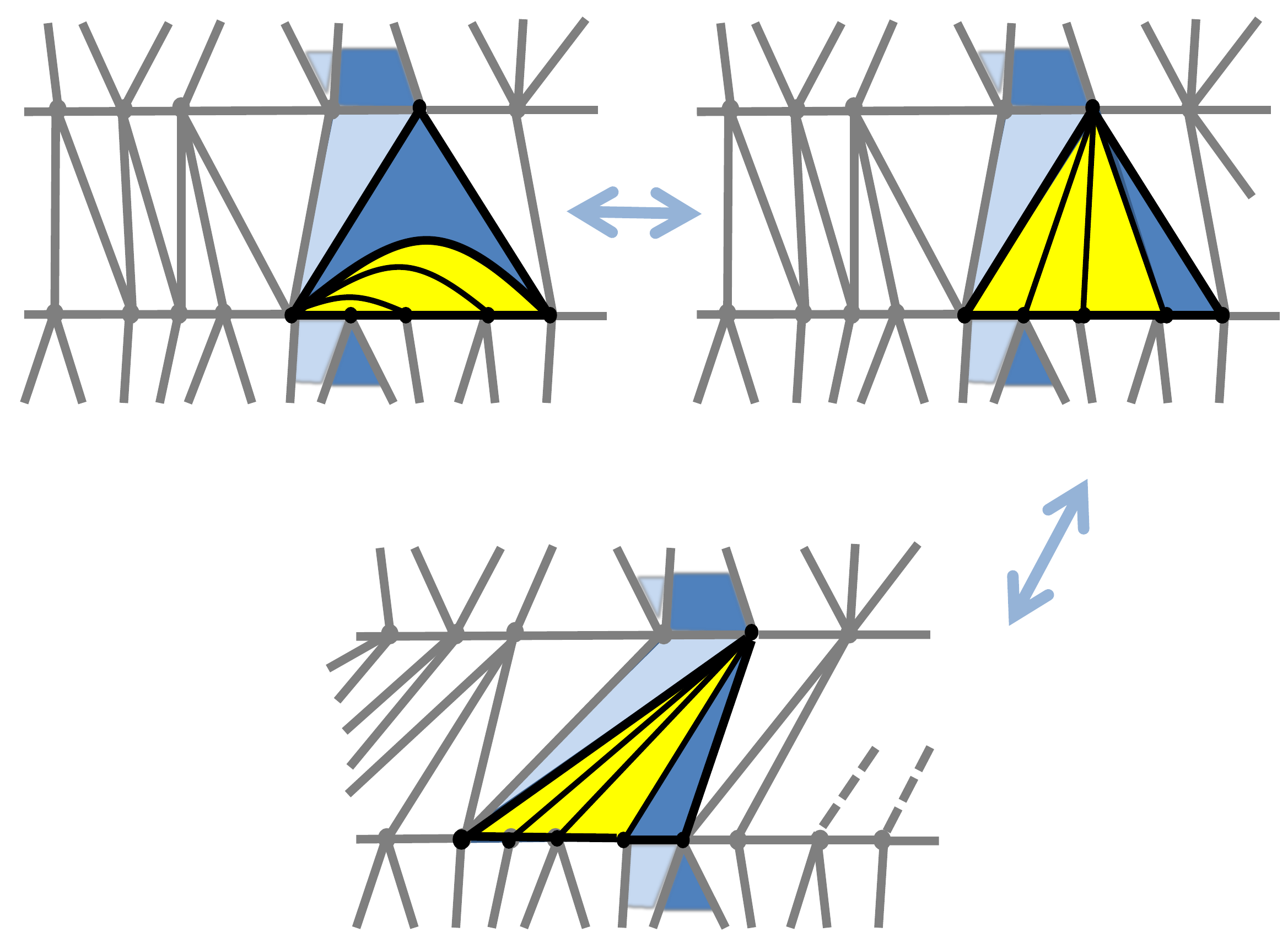}
\end{center}
\caption{Removing the ``defects'': Bijection between the set of almost causal triangulations created by the growth process and causal triangulations.}%
\label{fig5}
\end{figure}

In this section we have presented a growth process which samples almost causal triangulations by adding one triangle at a time using two different moves with probabilities \eqref{yamb-tp}.

The ``defects", where the triangulation generated by the growth process differs from a causal triangulation, can only occur at the moments where one starts a new strip and in particular if the process starts a strip by a sequence of $(-)$-moves. The defects can never occur during the filling and finishing of a strip (see Figure~\ref{f3}). For example starting the strip with a sequence of $k$ $(-)$-moves and then a $(+)$-move we obtain a configuration like in the left-up-side picture in Figure~\ref{fig5}. In fact, one can transform any such almost causal triangulation to a causal triangulation.

%
\begin{lm}\label{bijection}
There is a one-to-one map between $\mathcal{C}^g_{m_0}(t)$, the set of all possible triangulations
$T_{n_t}$ formed by the subset of almost causal triangulations created by the growth process started from $\mathcal{S}_{m_{0}}$ and stopped at $n_t$, i.e. where all vertices of the boundary are at distance $t$ to the $0$-root, and the set $\mathcal{C}^{c}_{m_0}(t)$ of rooted causal triangulations of height $t$ with initial boundary of length $m_0$.
\end{lm}

\begin{proof}
We defined $\mathcal{C}^g_{m_0}(t)$ to be the set of all possible rooted triangulations of the disc with all vertices of the boundary at distance $t$ to the $0$-root obtained by applying permitted sequences of the $(+)$ and $(-)$ moves starting from $\mathcal S_{m_0}$. 
Note that the set $\mathcal{C}^g_{m_0}(t)$ is only a subset of the set of all possible almost causal triangulations allowed by its definition. We denoted $\mathcal{C}^{c}_{m_0}(t)$ to be the set of all rooted causal triangulation of the cylinder $C_t$, with $m$ vertices on the zero-slice. Lemma~\ref{bijection} then states that there exists an one-to-one correspondence between  $\mathcal{C}^g_{m_0}(t)$ and $\mathcal{C}^{c}_{m_0}(t)$. To prove the Lemma we give an explicit construction of this correspondence:

The construction is divided in two steps as illustrated in Figure~\ref{fig5}. The first step is to ``correct" the sequence of moves that creates ``defects" in the triangulations. Suppose that we have defects in the $i$-th strip $S^1\times [i,i+1]$.
The move that describes the transformation of such an almost causal triangulation into a causal triangulation is presented in the upper line of  Figure~\ref{fig5}, where the sequence of moves $$\underbrace{(-),\dots, (-)}_{l}(+)$$ is substituted by the sequence of moves $$(+)\underbrace{(-),\dots, (-)}_{l }$$ or in terms of the $M_n$ process:
\begin{eqnarray}
&& \{\xi_{n_i+1}=-1, ... , \xi_{n_i+l}=-1,\xi_{n_i+l+1}=+1 \} \to \nonumber \\
 &&\,\,\,\,\,\,\,\,\,\,\,\,\,\,\, \,\,\,\,\,\,\,\,\,\,\,\,\,\,\,  \{\xi_{n_i+1}=+1, \xi_{n_t+2}=-1, ... , \xi_{n_i+l+1}=-1 \}.\nn
\end{eqnarray}
One can further verify that
\begin{eqnarray}
&& \P( \xi_{n_i+1}=-1, ... , \xi_{n_i+l}=-1,\xi_{n_i+l+1}=+1 |M_{n_i}=m) = \nonumber \\
 &&\,\,\,\,\,\,\,\,\,\,\,\,\,\,\, \,\,\,\,  \P( \xi_{n_i+1}=+1, \xi_{n_i+2}=-1, ... , \xi_{n_i+l+1}=-1|M_{n_i}=m ). \label{Pequal}
\end{eqnarray}

The second step is a {\it $l$}-shift of the $i$-th slice of the strip $S\times\{i\}$. Suppose $v_1^i, v_2^i, \dots, v_m^i$ is the sequence of vertices in $i$-th slice and $v_1^{i+1}, v_2^{i+1}, \dots, v_n^{i+1}$ on the $(i+1)$-th slice of the strip. After the first step the triangulation on $S\times\{i\}$ is a causal triangulation. The $l$-th shift is defined by the following transformation of the strip: any vertex $v_j^i$ in the $i$-th slice is shifted to the vertex $v_{j-l}^i$, i.e. any ``up" triangle $(v_j^i, v_{j+1}^i, v_k^{i+1})$ becomes a $(v_{j-l}^i, v_{j+1-l}^i, v_k^{i+1})$ triangle and any ``down" triangle $(v_{j^\prime}^i, v_{k^\prime}^{i+1}, v_{k^\prime+1}^{i+1})$ becomes a $(v_{j^\prime-l}^i, v_{k^\prime}^{i+1}, v_{k^\prime+1}^{i+1})$ triangle, where the sum in indices is the cyclic sum modulo $m$.

Note that the first step gives us already a causal triangulation. However, the marked edges can be anywhere in the slice. The second step deals with this problem. After the second step all dark blue triangles are connected to each other as marked triangles on the causal triangulations.

The inverse transformation is clear now. Consider the marked triangle of a causal triangulation in a strip $S\times [i,i+1]$. If the left-neighbor triangle is a ``down" triangle we do not change anything in this strip. Otherwise  we find the first left ``down" triangle $\Delta$, thus all ``up" triangles between the marked triangle and $\Delta$ are yellow triangles in the Figure 5. After that the transformation is obvious. This shows that there is an one-to-one correspondence between  $\mathcal{C}^g_{m_0}(t)$ and $\mathcal{C}^{c}_{m_0}(t)$ which completes the proof.
\end{proof}

We now have the following theorem:

\begin{theo} \label{thmbijUICT}
The growth process $M_n$ with $M_0=m_0$ stopped at times $n_t$ samples a section $T\in\mathcal{C}^c_{m_0}(t)$ of a UICT of height $t$ with initial boundary of length $m_0$, where $\mathcal{C}^c_{m_0}(t)$ is the image of $\mathcal{C}^g_{m_0}(t)$ under the transformation described in the proof of Lemma~\ref{bijection}.
\end{theo}

\begin{proof}
By the bijection of Lemma~\ref{bijection}, the growth process $M_n$ with $M_0=m_0$ stopped at times $n_t$ creates every possible causal triangulation of height $t$ with initial boundary of length $m$. Further, by \eqref{Pequal} removing the defects does not affect the probabilities. One has that the probability of each vertex at height $i$ having $k$ down triangles in $[i,i+1]$ attached to it is $p_k=1/2^{k+1}$.
One can also show that the probability to create a new slice with boundary length $m+k$ given that the preceding  boundary is of length $m$, is equal to the corresponding probability for the UICT. Indeed, we calculate the probability $\mathbb P (M_{n_{i+1}} = m+k \mid M_{n_i}=m)$ in Lemma \ref{lm:probstrip} later, yielding
\begin{equation}\label{y-prob}
 P (M_{n_{i+1}} = m+k \mid M_{n_i}=m) = \frac{m+k}{m} \frac{1}{2^{2m+k}} \binom{2m+k-1}{m-1}.\nn
\end{equation}
This completes the proof.
\end{proof}

\begin{rmk} One observes that in
\begin{equation}
     \mathbb P( (\pm)-\mbox{move} \mid l(T_n)=m ) = \frac{1}{2} \frac{m\pm 1}{m}.\nn
\end{equation}
the factor $1/2$ directly relates to the off-spring distribution $p_k=2^{-k-1}$, while the pre-factor, $(m\pm 1)/m$ results in the conditioning of the branching process on non-extinction.
\end{rmk}

\section{Growth rate and fractal dimension}\label{rateSec}

In this section we are interested in determining the growth rate of the process $\{M_n\}_{n\geq 0}$ and from this the fractal dimension $d_h$ of the (infinite) causal triangulation generated by this process.

Recall the definition of the moments $n_t, t=1, 2, \dots$ in Definition \ref{yamb-mom1}. By Remark \ref{yamb-th1} we can also obtain the moments $n_t$ from the Markov chain $\{M_n\}_{n\geq 0}$. 
%
%
\begin{lm}
\label{lm:defeq}
Let $n_1=0$. Suppose that $n_{t-1}$ is defined, then by Definition \ref{yamb-mom1} one has
\begin{eqnarray}
n_t&=&\min\{s: s-n_{t-1}=M_{s}+M_{n_{t-1}}\},\label{e:df1}\\
%
%
n_t&=&\min\Bigl\{s: \text{$s> n_{t-1}$ and $\sum_{k=n_{t-1}}^{s-1} \mathrm I \{\xi_k=-1\}=M_{n_{t-1}}$}\Bigr\}.\label{e:df2}
\end{eqnarray}
%
\end{lm}

\begin{proof}
The proof is a direct consequence of Definition \ref{yamb-mom1} of the stopping times $\{n_t\}_{t=1,2, \dots}$. Let us suppose that $M_{n_{t-1}} = m$. Hence, to complete the next slice, we need to put $m$ ``up" triangles, considering that only the first ``up" triangle we add using the $(+)$-move, while the remaining ``up" triangles are added using $(-)$-moves. Note also that we complete the slice with the last $(-)$-move which adds a ``down" triangle. Thus, to fill the slice we  need exactly $m$ $(-)$-moves, and some (random) number of $(+)$-moves. This is provided exactly by \eqref{e:df2}: $n_t$ is the growth time of the $m$-th $(-)$-move after $n_{t-1}$. This shows that Definition \ref{yamb-mom1} implies \eqref{e:df2}.

We now show that \eqref{e:df1} is equivalent to \eqref{e:df2}. Equation~\eqref{e:df1} reflects the fact that the moment $n_t$ is the hitting time of the straight line $M_s=(s-n_{t-1})-M_{n_{t-1}}$ (the line $bdf$  in Figure~\ref{yambar-fig1}). To hit this line the process needs exactly $m$ steps with $\xi\!=\!-1$.  This proves the equivalence of the two equations.
\end{proof}

Let us consider a causal triangulation generated by the growth process $\{M_n\}_{n\geq 0}$ with initial boundary of length $M_0=1$ using the probabilities \eqref{yamb-tp1}. Let $\Gamma(t)$ denote the set of triangles of the corresponding triangulation with all vertices having graph distance less or equal than $t$ from the initial boundary. The fractal dimension $d_h$ describes the growth of $|\Gamma(t)|\sim t^{d_h}$ as $t\to\infty$. Observing that at instance $n_t$ of the growth process we have $|\Gamma(t)|=n_t$, we define the fractal dimension as the limit (if it exists)
\begin{equation}
d_h = \lim_{t\to\infty} \frac{\log n_t}{\log t}.\nn
\end{equation}

\begin{figure}[t]
\begin{center}
\includegraphics[width=11cm]{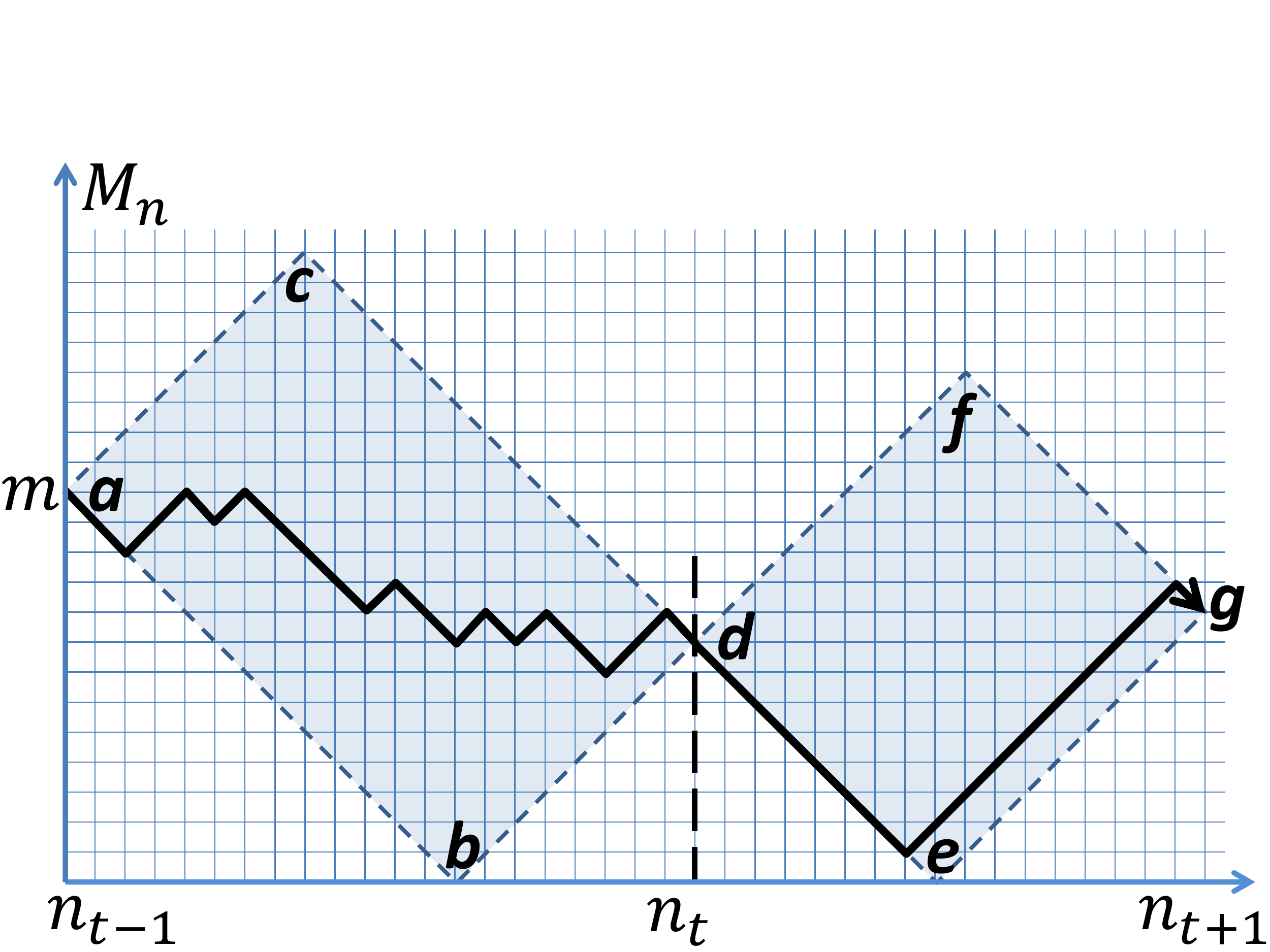}
\end{center}
\caption{Illustration of a path of the Markov chain $\{M\}_{n=0,1,2,...}$.}%
\label{yambar-fig1}
\end{figure}

We now want to prove that the fractal dimension of an (infinite) causal triangulation generated by the growth process $\{M_n\}_{n\geq 0}$ with probabilities \eqref{yamb-tp1} is almost surely 2. To do so we first prove the following slightly stronger statement:

\begin{prop}\label{theo:nt}
For almost every trajectory $\omega$ of the growth process $\{M_n\}_{n\geq 0}$ with probabilities \eqref{yamb-tp1} and with initial boundary of length $M_0=1$, there exist two constants $\gamma_1=\gamma_1(\omega)>0$ and $\gamma_2=\gamma_2(\omega)>0$ such that
\begin{equation}\label{in1}
    \gamma_1\frac{t^2}{\log^2 t} \le n_t \le \gamma_2 t^2 \log^2 t 
\end{equation}
\end{prop}
Note that in the proposition the constants $\gamma_1$ and $\gamma_2$ depend on the whole trajectory $\omega$ of the process, but not on $t$. Using the definition of the fractal dimension we then have the following Theorem:

\begin{theo}\label{theo:fractal}
An infinite causal triangulation generated by the growth process $\{M_n\}_{n\geq 0}$ with probabilities \eqref{yamb-tp1} with initial boundary of length $M_0=1$ has fractal dimension $d_h=2$ almost surely.
\end{theo}

\begin{proof} The proof follows immediately from the previous Proposition \ref{theo:nt}.
\end{proof}

This Theorem is analogous to a result by Durhuus, Jonsson and Wheater (Theorem 3 in \cite{Durhuus:2009sm}) which is derived for UICT using the bijection to critical Galton-Watson processes conditioned to never die out. In this construction, subsequent generations in the branching process correspond to vertices of subsequent slices of fixed minimal graph distance, i.e.\ geodesic distance, from the initial boundary. This is in contrast to the construction through the growth process where the triangulation is grown triangle by triangle. One can think of both constructions as being dual to each other in the sense that in the branching process picture geodesic distance is fixed and area growth is estimated while in the growth process area is fixed and geodesic distance is estimated. We will comment further on this duality in Section \ref{conclusion}.

To prove Proposition~\ref{theo:nt} we need three lemmas of which the first two are proven in Appendix \ref{appA1}.
\begin{lm}
We have
\label{lm:sn21}
\begin{equation}
\frac{M_n}{\sqrt{n\log n}} \to 0 \quad \text{a.s.}\nn
\end{equation}
\end{lm}

\begin{lm}
We have
\label{lm:sn22}
\begin{equation}
\frac{M_n-\sum_{i=0}^n \frac{1}{M_i}}{\sqrt{n}\log n} \to 0 \quad \text{a.s.}\nn
\end{equation}
\end{lm}

\begin{lm}
\label{lm:sn23}
For any $t$, we have
\begin{equation}
\max_{n_{t}<i\le n_{t+1}} M_i \le n_{t+1}-n_{t} .\nn
\end{equation}
\end{lm}

\begin{proof}[Proof of Lemma~\ref{lm:sn23}]
The proof follows directly from the Figure~\ref{yambar-fig1}. For any given $M_{n_{t}}$, $M_{n_{t+1}}$, $n_{t}$ and $n_{t+1}$ the trajectory of the growth process belongs to the rectangle $defg$ in Figure~\ref{yambar-fig1}. This means that the maximal accessible point is $f$. Its $y$-coordinate is equal to $M_{n_{t}}+M_{n_{t+1}}$. Thus
\[
\max_{n_{t}<i\leq  n_{t+1}} M_i \leq M_{n_{t}}+M_{n_{t+1}},
\]
for any trajectory of the growth process.
By \eqref{e:df1} of Lemma \ref{lm:defeq}, the RHS is equal to $n_{t+1}-n_t$. This proves the Lemma.
\end{proof}

\begin{proof}[Proof of Proposition~\ref{theo:nt}] 
First let us prove the RHS of (\ref{in1}).
From Lemma~\ref{lm:sn21} it follows that for almost every trajectory $\omega$ of the growth process
there exists a constant $C_1= C_1(\omega) >0$ (in the following we will omit $\omega$ from our notations) such that, for any $n$,
\begin{equation}\label{in2}
    M_n \le C_1 \sqrt{n\log n}
\end{equation}
Using \eqref{e:df1} and \eqref{in2}, we get
\begin{eqnarray}
n_t&=&n_t-n_0=\sum_{k=0}^{t-1} ( n_{k+1}-n_k )
=\sum_{k=0}^{t-1} (M_{n_{k+1}}+M_{n_k} )\nonumber\\
&\le& 2\sum_{k=1}^{t} M_{n_k}
\le 2C_1\sum_{k=1}^{t} \sqrt{n_k \log n_k}
\le 2C_1t\sqrt{n_t \log n_t}.\label{in3}
\end{eqnarray}
From this and using the inequality $\log n< \sqrt{n}$ we also have
\[
4C_1^2t^2\ge \frac{n_t}{\log n_t} \ge \frac{n_t}{ \sqrt{n_t} } = \sqrt{n_t}
\]
and thus,
\begin{equation}\label{in31}
n_t \le 16C_1^4t^4.
\end{equation}
From \eqref{in3} and \eqref{in31} it follows that there exists a constant $\gamma >0$ such that
\[
n_t \le  4C_1^2t^2\log n_t \le \gamma t^2 \log t,
\]
where we use \eqref{in3} for the first inequality and \eqref{in31} for the second one. Note also that the constant $\gamma$ depends on the whole trajectory of the process, but does not depend on $t$. Thus, the RHS of (\ref{in1}) is proved.

Now let us prove the LHS of  (\ref{in1}). 
On the one hand, using Lemma~\ref{lm:sn23}, one has
\begin{equation}\label{in5}
\begin{split}
\sum_{i=n_1}^{n_t} \frac{1}{M_i}
&\ge
\sum_{k=1}^{t-1}\sum_{i=n_k+1}^{n_{k+1}} \frac{1}{M_i}\\
&\ge
\sum_{k=1}^{t-1} \frac{n_{k+1}-n_k}{\max_{n_k<i\le n_{k+1}}M_i}
\ge
\sum_{k=1}^{t-1} \frac{n_{k+1}-n_k}{n_{k+1}-n_k}
\ge t .
\end{split}
\end{equation}
On the other hand, using Lemma~\ref{lm:sn22} it follows that there exists a $C_2= C_2(\omega)>0$ such that
$$
\sum_{i=n_1}^{n_t} \frac{1}{M_i} - M_{n_t}
\le C_2 \sqrt{n_t}{\log n_t},
$$
and from \eqref{in2} it follows that there exists a $C_3>0$ such that
\begin{equation}\label{in6}
\sum_{i=n_1}^{n_t} \frac{1}{M_i}
\le C_3 \sqrt{n_t}{\log n_t}.
\end{equation}
Combining \eqref{in5} and \eqref{in6}, we get
\begin{equation}\label{in7}
t \le C_3 \sqrt{n_t}{\log n_t}.
\end{equation}
Finally, lifting the square in \eqref{in7} and using \eqref{in31}, we get the LHS of (\ref{in1}).
\end{proof}

\begin{rmk} \label{rmktime}
From the proof of Proposition~\ref{theo:nt} we observe that
\[
t\leq \sum_{i=1}^{n_t} \frac{1}{M_i}.
\]
In fact, as we will see in the next section, the properly rescaled inverse of $n_t$, i.e.\ $t_n := \min \{s: n_s \geq n \}$ and $\sum_{i=1}^{n}1/M_i$ converge to the same limiting process $T_u = \int_0^u \frac{1}{2M_s} ds$.
\end{rmk}

\section{Weak convergence results}\label{convergenceSec}

Consider the rescaled process
\begin{equation}
M^{(n)}_{u}=\frac{M_{[un]}}{\sqrt{n}}, \quad 0\leq u <\infty.\nn
\end{equation}

\begin{theo} \label{thmlength}
For an infinite causal triangulation generated by the growth process $\{M_n\}_{n\geq 1}$ with initial boundary of length $M_0/\sqrt{n} \to l\geq 0$ we have
\[
M^{(n)}_{u} \codi M_u, \quad 0\leq u <\infty,
\]
as $n\to\infty$ in the sense of weak convergence on the functions space $D[0,\infty)$, where the continuous process is diffusive and solves the following It\^o's equation
\[
dM_u=\frac{1}{M_u}du + dB_u, \quad M_0=l,
\]
with $B_u$ being standard Brownian motion of variance one.
\end{theo}

\begin{proof}
Consider the process $M^{(n)}_{u}=M_{[un]}/\sqrt{n}$ as defined above. Using the notation of Appendix  \ref{appB2} with $h=1/n$, we now set $Y_{m/n}^{(n)}=M_{m}/\sqrt{n}$, where the state space is $\mathcal{X}_{1/n}=\{k/\sqrt{n},0<k\leq n\}$. From the transition probability of the Markov chain $M_n$ as discussed in Section \ref{growthSec} one has
\[
p^{1/n}(x,x+1/\sqrt{n})=\frac{1}{2}\frac{x+1/\sqrt{n}}{x}, \quad p^{1/n}(x,x-1/\sqrt{n})=\frac{1}{2}\frac{x-1/\sqrt{n}}{x},
\]
for any $x\in \mathcal X_{1/n}$.
We now want to apply Theorem \ref{convbook}. It is easy to first check condition $(3)$. From the transition probabilities one immediately sees that for $\epsilon>1/\sqrt{n}$, $\Delta_{1/n}^\epsilon (x)=0$ for all $x$ and thus condition $(3)$ holds. Further,
\begin{eqnarray}
b_{1/n}(x)&=& \frac{n}{2} \left(\frac{1}{\sqrt{n}} \frac{x+1/\sqrt{n}}{x}-\frac{1}{\sqrt{n}} \frac{x-1/\sqrt{n}}{x} \right)= 1/x  \nonumber\\
\sigma^2_{1/n}(x)&=& \frac{n}{2} \left(\frac{1}{n} \frac{x+1/\sqrt{n}}{x}+\frac{1}{n} \frac{x-1/\sqrt{n}}{x} \right)= 1  \nonumber
\end{eqnarray}
This shows that the conditions $(1)$ and $(2)$ hold for $b(x)=1/x$ and $\sigma^2(x)=1$. Further,  one observes that $MP({1/x},{1})$ is a well-posed martingale problem (see Definition \ref{MPdef}). Hence by Theorem \ref{convbook}, setting $X_t^{(n)}= Y^{(n)}_{[nt]/n}=M_{[t n]}/\sqrt{n}$ one obtains the desired result.
\end{proof}

We have shown convergence of the rescaled boundary of our growth process to a limiting diffusion. Here, the diffusion time $u$ is the time associated to the growth process. In the following we would like to relate this to the corresponding diffusion process, where the diffusion time is geodesic distance. Consider therefore the rescaled process
\begin{equation}
L^{(t)}_{s}=\frac{M_{n_{[s t]}}}{t}, \quad 0\leq s <\infty,\nn
\end{equation}
where $n_t$ is defined as above. We then have the following theorem:

\begin{theo} \label{thmlength2}
For an infinite causal triangulation generated by the growth process $\{M_n\}_{n\geq 1}$ with initial boundary of length $M_1 \to l\geq 0$ we have
\[
L^{(t)}_{s}=\frac{M_{n_{[s t]}}}{t} \codi L_s, \quad 0\leq s <\infty,
\]
as $t\to\infty$ in the sense of weak convergence on the functions space $D[0,\infty)$, where the continuous process is diffusive and solves the following It\^o's equation
\[
dL_s =2 ds +  \sqrt{2 L_s} dB_s \quad L_0=l,
\]
with $B_s$ being standard Brownian motion of variance one.
\end{theo}

The proof is based on the following lemma:

\begin{lm}
\label{lm:probstrip}
We have
\[
\P\left(M_{n_{t+1}}=m+k | M_{n_{t}}=m\right) = \frac{m+k}{m} \frac{1}{2^{2m+k}} \binom{2m+k-1}{m-1},
\]
where $k=-m+1, -m+2, \dots.$
\end{lm}

\begin{proof} By Lemma \ref{lm:defeq}  we have that $n_{t+1}-n_t=2m+k$. Recall that any trajectory of the growth process needs a ``$-1$" step (i.e.\ $(-)$-move) to finish a slice, i.e.\ in Figure~\ref{yambar-fig1} the trajectory hits the point $d$ ``from above". For given $M_{n_{t+1}}=m+k$ and $M_{n_{t}}=m$ we have $\binom{2m+k-1}{m-1}$ possible trajectories of the growth process. The last observation is that any such trajectory, starting at the point $m$ and finishing at $m+k$ has the same probability $\frac{m+k}{m} \frac{1}{2^{2m+k}}$ (including the trajectories with reflection). This completes the proof.
\end{proof}

\begin{proof}[Proof of Theorem \ref{thmlength2}]
Using Lemma \ref{lm:probstrip} the transition probability from $M_{n_{t}}$ to $M_{n_{t+1}}$ is exactly equal the corresponding expression for the UICT, i.e.\  \eqref{sootn}. The proof then follows directly from Theorem 4.1 of \cite{Sisko2011}.
\end{proof}

\begin{theo} \label{thmlength12}
The diffusion process $M_u$, given in Theorem \ref{thmlength}, and the diffusion process $L_s$, given in Theorem \ref{thmlength2}, are related by a random time change $L_s=M_{T^{-1}_s}$ with
\[
T^{-1}_s=\inf\{ t: T_t>s\}
\]
and
\[
T_u = \int_0^u \frac{1}{2M_s} ds.
\]
\end{theo}
\begin{proof} The proof is a consequence of Theorem \ref{thmtimeapp}, a well-known theorem from stochastic calculus, presented in Appendix \ref{appB2} with $g(x)=1/(2x)$.
\end{proof}

\section{Discussion}\label{conclusion}

In this article we present a growth process which samples sections of uniform infinite causal triangulations (UICT). In particular, a triangulation is grown by adding a single triangle according to two different moves, denoted ``$(+)$" and ``$(-)$", with probability
\begin{equation}
    \mathbb P( (\pm)-\mbox{move} \mid l(T_n)=m ) = \frac{1}{2} \frac{m\pm 1}{m},\nn
\end{equation}
where $l(T_n)$ denotes the length of the boundary of the triangulation at step $n$ of the growth process. The $(\pm)$-moves are illustrated in Figure \ref{fig1}, the $(+)$-move increases the boundary length by one while the $(-)$-move decreases it by one.
This growth process can equivalently be described by a recurrent Markov chain $\{M_n\}_{n\geq0}$ for the boundary length of the triangulation $M_n=l(T_n)$ as noted in Remark \ref{yamb-th1}.

It is shown in Theorem \ref{yamb-th2} that the growth process constructs so-called almost causal triangulations which are causal triangulations with certain defects as shown in Figure \ref{f3}. Defining the stoping times $n_t$ when the growth process completes the strip $S^1\times [t,t+1]$ at ``height'' $t$, it is shown in Theorem \ref{thmbijUICT} that there is a bijection between the almost causal triangulations created by the growth process and (regular) causal triangulations which furthermore preserves the probability of the corresponding triangulation. Hence, the growth process $\{M_n\}_{n\geq0}$ with $M_0=m_0$ when stopped at $n_t$ indeed samples sections  of UICT of height $t$ with initial boundary equal to $m_0$ and final boundary arbitrary.

Using the growth process, as described above, we show in Proposition \ref{theo:nt} that for almost every trajectory of the growth process one can find two constants $\gamma_1,\gamma_2>0$ such that $  \gamma_1\, t^2/(\log^2 t) \le n_t \le \gamma_2 t^2 \log^2 t $. This implies that the fractal dimension is given by $d_h=2$ almost surely as stated in Theorem \ref{theo:fractal}. This derivation is dual to previous results in \cite{Durhuus:2009sm} which employs the relation to branching processes: In the branching process geodesic distance is fixed while area growth is estimated, whereas in the growth process area (which is equal to growth time) is fixed and geodesic distance is estimated.

\begin{table}[t]
\begin{center}
\begin{tabular}{| l ||c|c|c|}
\hline
& boundary length & area & geodesic distance \\\hline\hline
Growth process & $M_u$ & $0\leq u< \infty$ & $T_u = \int_0^u \frac{1}{2 M_v} dv$ \\\hline
Branching process& $L_s$ & $A_s = \int_0^s 2L_v dv$ & $0\leq s<\infty$ \\\hline
\end{tabular}
\end{center}
\caption{Duality of growth process and branching process.\label{Tab1}}
\label{default}
\end{table}%

In Theorem \ref{thmlength} we discuss convergence of the rescaled Markov chain $\{M_n\}_{n\geq0}$ to a diffusion process given by the following It\^o's equation
\[
dM_u=\frac{1}{M_u}du + dB_u.
\]
It is then shown in Theorem \ref{thmlength2} and \ref{thmlength12} that by a random time change $L_s=M_{T^{-1}_s}$ one obtains a diffusion process with It\^o's equation
\[
dL_s =2 ds +  \sqrt{2 L_s} dB_s,
\]
which describes the behaviour of the boundary length of completed slices and precisely agrees with the corresponding results obtained from the branching process picture \cite{Sisko2011}. While Theorem \ref{thmlength} follows rather straightforwardly from the properties of the Markov chain, Theorem \ref{thmlength2} together with Theorem \ref{thmlength12} result in a physically interesting duality relation which is illustrated in Table \ref{Tab1}: In the growth process growth time $u$ which is equal to area is fixed and geodesic distance $T_u = \int_0^u \frac{1}{2 M_v} dv$ is random, whereas, in the branching process picture geodesic distance $s$ is fixed while area $A_s = \int_0^s 2L_v dv$ is random. This duality relation also clarifies the so-called peeling procedure as introduced by Watabiki \cite{Watabiki:1993ym} in the context of Euclidean quantum gravity and derived in \cite{CDTmatrix2} for CDT.

As a continuation of the presented work it would be interesting to extend the work of Angel \cite{Angel:2002ta} and investigate the convergence of the boundary length process coming from the growth process of DT to a L\'evy process. Furthermore, one could extend the here developed techniques to multi-critical DT \cite{Watabiki:1993ym,Gubser:1993vx} as well as to a recently introduced model of multi-critical CDT \cite{Ambjorn:2012zx,Atkin:2012yt}. The convergence of the boundary length process should hopefully shed light on the failure of the peeling procedure in the context of multi-critical DT \cite{Watabiki:1993ym,Gubser:1993vx}.

\subsection*{Acknowledgments}
The authors would like to thank Richard Gill for fruitful discussions. The work of V.S.\ was supported by FAPERJ (grants E-26/170.008/2008 and E-26/110.982/2008) and CNPq (grants 471891/2006-1, 309397/2008-1 and 471946/2008-7). The work of A.Y.\ was partly supported by CNPq 308510/2010-0. S.Z.\ would like to thank the Department of Statistics at S\~ao Paulo University (IME-USP) as well as the Institute for Pure and Applied Mathematics (IMPA) for kind hospitality. Financial support of FAPESP under project 2010/05891-2, as well as STFC and EPSRC is kindly acknowledged.

\appendix

\section{Proofs of basic Lemmas} \label{appA1}

\subsection{Proof of Lemma \ref{lm:sn21}}

Let $\F_n=\sigma(M_0,M_1,\dots,M_n)$. Recall that
$\xi_n=M_{n+1}-M_n$.
We thus have $M_n \in \F_n$ and $\xi_n \in \F_{n+1}$.
On $\{M_n \geq 1\}$ one has
\begin{equation}
\label{e:lm:sn2:1}
\E(\xi_n|\F_n)
=1\cdot \frac{1}{2}\Bigl(1+\frac{1}{M_n}\Bigr)-1\cdot
\frac{1}{2}\Bigl(1-\frac{1}{M_n}\Bigr)=\frac{1}{M_n}.
\end{equation}
Note also that
\begin{equation}
\label{e:lm:sn2:3}
\xi_n^2=1.
\end{equation}

Consider $X_n=M_n^2-3n$.
Let us prove that $X_n$ is a martingale adapted to $\F_n$.
Evidently $X_n \in \F_n$ and $\E |X_n|< \infty$.  Therefore, we only
need to check that $\E(X_{n+1}|\F_n)=X_n$.
Using~\eqref{e:lm:sn2:1}--\eqref{e:lm:sn2:3} we have
\begin{equation}
\label{e:lm:sn2:4}
\begin{split}
&\E(X_{n+1}-X_n|\F_n)
=\E(M_{n+1}^2-M_n^2-3|\F_n)\\
&\quad=\E[(\xi_n+M_n)^2-M_n^2-3|\F_n]
=\E[\xi_n^2|\F_n]
+2M_n\E[\xi_n|\F_n]-3
=0.
\end{split} \nn
\end{equation}
Thus one gets
\begin{eqnarray}
\label{e:lm:sn2:5}
\E((X_{n+1}-X_n)^2|\F_n)
&=&\E[(\xi_n^2 +2M_n\xi_n-3)^2|\F_n] \nn\\
&=& 4\E[(M_n\xi_n-1)^2|\F_n]\nn\\
&=& 4(M^2_n+1-2)=4(M_n^2-1),\nn
\end{eqnarray}
and therefore,
\begin{equation}
\label{e:lm:sn2:6}
\E(X_{n+1}-X_n)^2
=4\E M_n^2 -4
=4\E (X_n+3n)-4
=4\E (X_0) +12n-4.
\end{equation}
With any sequence of positive numbers $\{a_m\}$ consider
\[
B_n=\sum_{m=1}^n \frac{X_m-X_{m-1}}{a_m}.
\]
Using the fact that $X_n$ is a martingale, it is easy to check that $B_n$ is
also a martingale.
One has
\begin{eqnarray}
\E B_n^2&=&\E\biggl(\E\biggl(\biggl(\sum_{m=1}^n
 \frac{X_m-X_{m-1}}{a_m}\biggr)^2\biggl|\biggr.\F_{n-1}\biggr)\biggr)\nn\\
 &=&\E\biggl(\sum_{m=1}^{n-1} \frac{X_m-X_{m-1}}{a_m}\biggr)^2 + \E\biggl(\frac{X_n-X_{n-1}}{a_n}\biggr)^2
\nn\\
&&\ \qquad + 2 \E\biggl( \sum_{m=1}^{n-1} \frac{X_m-X_{m-1}}{a_m}
\E\biggl(\frac{X_n-X_{n-1}}{a_n}\biggl|\biggr.\F_{n-1}\biggr)\biggr) \nn\\
& =&\E\biggl(\sum_{m=1}^{n-1} \frac{X_m-X_{m-1}}{a_m}\biggr)^2
+\E\biggl(\frac{X_n-X_{n-1}}{a_n}\biggr)^2
=\dotsb\nn\\
&=&\sum_{m=1}^n \frac{\E (X_m-X_{m-1})^2}{a_m^2}.\nn
\end{eqnarray}
From \eqref{e:lm:sn2:6}, we get that if $a_m=m\log m$, then
there exists a constant $C>0$ such that $\E B_n^2<C<\infty$.
Therefore, using the $L^p$ convergence theorem (see e.g.\ Theorem (4.5) from Chapter 4 of \cite{Durrett1996}), 
one has that $B_n$ converges a.s. Using Kroneker's Lemma (see e.g.\ Lemma (8.5) from Chapter 1 of \cite{Durrett1996}), 
we see that $\frac{X_n}{n \log n}\to 0$ a.s.
and recalling the definition of $X_n$, one gets Lemma
\ref{lm:sn21}.

\subsection{Proof of Lemma \ref{lm:sn22}}

This proof proceeds using a similar strategy as the previous proof.
%
%
Consider
\[
X_n=M_n-\sum_{i=1}^n \frac{1}{M_i}.
\]
From~\eqref{e:lm:sn2:1} it follows that
$X_n$ is a martingale adapted to $\F_n$.
Using the fact that $|\xi_n|=1$ and $M_n\ge 1$ for any $n$,
we have
\begin{equation}
\label{e:lm:sn2:5b}
\E(X_{n+1}-X_n)^2
=\E\Bigl(\xi_n-\frac{1}{M_{n+1}}\Bigr)^2
\le 4.
\end{equation}
Consider further
\[
B_n=\sum_{m=1}^n \frac{X_m-X_{m-1}}{a_m}.
\]
Using the fact that $X_n$ is a martingale, one has again that $B_n$ is
also a martingale and 
we have
\begin{eqnarray}
\E B_n^2&=& 
\sum_{m=1}^n \frac{\E (X_m-X_{m-1})^2}{a_m^2}.\nonumber
\end{eqnarray}
Using~\eqref{e:lm:sn2:5b}, we get that if $a_m=\sqrt{m}\log m$, then
$\sup_n\E B_n^2<\infty$.
Using the $L^p$ convergence theorem, 
we observe that $B_n$ converges a.s. Finally,
using Kroneker's Lemma as before, we see that $\frac{X_n}{\sqrt{n} \log n}\to 0$ a.s.
Recalling the definition of $X_n$, we get  Lemma \ref{lm:sn22}.







\section{Convergence of Markov chains to diffusion processes}\label{appB2}

To prove Theorem \ref{thmlength} we need a little background on stochastic differential equations and convergence to diffusion. The following definition and theorem can for instance be found in \cite{Durrettstochasticbook} Chapter 5 and 8, where the latter is a rather good introduction to the topic which itself is based on \cite{Billingsley1999,Either1986,Strook1979}.

\begin{df} We say that $X_t$ is a solution to the martingale problem for $b$ and $\sigma^2$, or simply $X$ solves $\MP{b}{\sigma^2}$ if
\[
X_t -\int_0^t b(X_s) ds\quad\text{and}\quad  X_t^2 -\int_0^t \sigma(X_s)^2 ds
\]
are local martingales. Further, we say that the martingale problem is well-possed if there is uniqueness in distribution and no explosion. \label{MPdef}
\end{df}

Let us now consider a Markov chain $Y^{(h)}_{mh}$, $m\geq 0$, taking values in a set $\mathcal{X}_h\subset \mathbb{R}$ and having transition probabilities
\[
p^h(x,A):=\Pc{Y^{(h)}_{(m+1)h}\in A}{Y^{(h)}_{mh}=x}, \quad x\in \mathcal{X}_h,\quad A\subseteq\mathbb{R}.
\]
Further, set $X_t^{(h)}=Y^{(h)}_{h [t/h]}$ and define
\begin{eqnarray}
\sigma^2_h(x) &=&h^{-1} \int_{|y-x|\leq1}(y-x)^2 p^h(x,dy) \nonumber\\
b_h(x) &=& h^{-1}\int_{|y-x|\leq1}(y-x) p^h(x,dy) \nonumber\\
\Delta_h^\epsilon (x) &=& h^{-1}\, p^h(x,B(x,\epsilon)^c) \nonumber
\end{eqnarray}
with $B(x,\epsilon)=\{y:|y-x|<\epsilon\}$.

The following Theorem (see e.g.\ \cite{Durrettstochasticbook}, Theorem 8.7.1) proves convergence of the Markov chain to a limiting diffusion:
\begin{theo}\label{convbook}
Suppose that $b$ and $\sigma$ are continuous functions for which the martingale problem is well-possed and for $R<\infty$ and $\epsilon>0$
\begin{enumerate}
\item[(1)] $\lim_{h\to0}\sup_{|x|\leq R} |\sigma^2_h(x)-\sigma^2(x) |=0$
\item[(2)] $\lim_{h\to0}\sup_{|x|\leq R} |b_h(x)-b(x) |=0$
\item[(3)] $\lim_{h\to0}\sup_{|x|\leq R} \Delta_h^\epsilon (x)=0$
\end{enumerate}
If $X^{(h)}_0\to x$ then we have $X^{(h)}_t \codi X_t$, in the sense of weak convergence on the functions space $D[0,\infty)$, where the continuous process $X_t$ is diffusive and solves the following It\^o's equation
\[
dX_t=b(X_t)dt + \sigma(X_t) dB_t, \quad X_0=x.
\]
\end{theo}

The following is a well-known theorem from stochastic calculus (see e.g. \cite{Durrettstochasticbook}, Theorem 5.6.1) regarding random time changes of a stochastic process:

\begin{theo}\label{thmtimeapp}
Let $X_u$ be a solution of the martingale problem $\MP{b}{\sigma^2}$ for $u\in [0,\infty)$, let $g$ be a positive function and suppose that for all $u\in [0,\infty)$
\[
\tau_u = \int_0^u g(X_s) ds<\infty.
\]
Define the inverse of $\tau_u$ by $\tau^{-1}_s=\inf\{ t: \tau_t>s\}$ and let $Y_s=X_{\tau^{-1}_s}$, then $Y_s$ is a solution of $\MP{b/g}{\sigma^2/g}$.
\end{theo}

\newpage
\addcontentsline{toc}{section}{References}
\providecommand{\href}[2]{#2}\begingroup\raggedright\endgroup

\end{document}